\documentclass[a4paper,12pt]{article}

\usepackage{color}
\usepackage{amssymb, amsmath}
\usepackage{hyperref}
\usepackage{tikz}
 \usepackage{enumitem}

\newcommand\EFFACE[1]{}

\newcommand\ESOK[1]{#1}

\newcommand\OR[1]{\overrightarrow{#1}}
\def\ORT{\OR{\Theta}}

\newtheorem{theorem}{Theorem}[section]
\newtheorem{proposition}[theorem]{Proposition}

\newtheorem{lemma}[theorem]{Lemma}

\newtheorem{observation}[theorem]{Observation}

\newenvironment{proof}{
\par
\noindent {\bf Proof.}\rm}{\mbox{}\hfill$\square$\par\vskip 3mm}

\newcommand\pcn{\chi_{\rho}}

\def\LL{\ \longleftarrow\ }
\def\RR{\ \longrightarrow\ }

\makeatletter
\let\@fnsymbol\@arabic
\makeatother

\begin{document}

\title{{\bf On the Packing Coloring of Undirected\\ and Oriented Generalized Theta Graphs}}

\author{Daouya LA\"{I}CHE~\thanks{Faculty of Mathematics, Laboratory L'IFORCE, University of Sciences and Technology
Houari Boumediene (USTHB), B.P.~32 El-Alia, Bab-Ezzouar, 16111 Algiers, Algeria.}
\and Isma BOUCHEMAKH~\footnotemark[1]
\and \'Eric SOPENA~\thanks{Univ. Bordeaux, LaBRI, UMR5800, F-33400 Talence, France.}~$^,$\thanks{CNRS, LaBRI, UMR5800, F-33400 Talence, France.}~$^,$\footnote{Corresponding author. Eric.Sopena@labri.fr.}
}

\maketitle

\abstract{
The packing chromatic number $\pcn(G)$ of an undirected (respectively oriented) graph $G$ is the smallest integer $k$ such that
its set of vertices $V(G)$ can be partitioned into $k$ disjoint subsets $V_1$, \ldots, $V_k$, in such a way that every two distinct vertices
in $V_i$ are at distance (respectively directed distance) greater than $i$ in $G$ for every $i$, $1\le i\le k$.
The generalized theta graph $\Theta_{\ell_1,\dots,\ell_p}$ consists in two end-vertices joined by $p\ge 2$ internally vertex-disjoint paths
with respective lengths $1\le \ell_1\le\dots\le \ell_p$.

We prove that the packing chromatic number of any undirected generalized theta graph
lies between 3 and $\max\{5,n_3+2\}$, where $n_3=|\{i\ /\ 1\le i\le p,\ \ell_i=3\}|$, and that both these bounds are tight.
We then characterize undirected generalized theta graphs with packing chromatic number $k$ for every $k\ge 3$.
We also prove that the packing chromatic number of any oriented generalized theta graph
lies between 2 and~5 and that both these bounds are tight.
}

\medskip


\medskip


\section{Introduction}

All the graphs we consider are \ESOK{simple and loopless}.
For an undirected graph $G$, we denote by $V(G)$ its set of vertices and by $E(G)$ its set of edges.
The {\em distance} $d_G(u,v)$, or simply $d(u,v)$ when $G$ is clear from the context, between
vertices $u$ and $v$ in $G$ is the length \ESOK{(number of edges)} of a shortest path joining $u$ and $v$.
The {\em diameter} 
of $G$ is the maximum distance between two vertices of $G$.
We denote by $P_n$, $n\ge 1$, the path of order $n$ and by $C_n$, $n\ge 3$, the cycle of order~$n$.

A {\em packing $k$-coloring} of $G$ is a mapping $\pi:V(G)\rightarrow\{1,\ldots,k\}$
such that, for every two distinct vertices $u$ and $v$,  \ESOK{$\pi(u)=\pi(v)=i$ implies $d(u,v)>i$}.
The {\em packing chromatic number} $\pcn(G)$ of $G$ is then the smallest $k$ such that
$G$ admits a packing $k$-coloring.
In other words, $\pcn(G)$ is the smallest integer $k$ such that
$V(G)$ can be partitioned into $k$ disjoint subsets $V_1$, \ldots, $V_k$, in such a way that every two vertices
in $V_i$ are at distance greater than $i$ in $G$ for every $i$, $1\le i\le k$.

This notion extends to digraphs in a natural way~\cite{LBS15}, by considering the \emph{(weak) directed distance}
between vertices, defined as the number of arcs in a shortest directed path linking these vertices, in either direction.

Packing coloring of undirected graphs has been introduced by Goddard, Hedetniemi, Hedetniemi, Harris and Rall~\cite{GHHHR03,GHHHR08} under the name {\em broadcast coloring}
and has been studied by several authors in recent years.
Several papers deal with the packing chromatic number of certain classes of undirected graphs such
as trees~\cite{ANT14,BKR07,GHHHR08,RBFK08,S04},
lattices~\cite{BKR07,EFHL10,FKL09,FR10,KV14,SH10},
Cartesian products~\cite{BKR07,FKL09,RBFK08},
distance graphs~\cite{EHL12,EHT14,T14} or hypercubes~\cite{GHHHR08,TVP15,WRR14b}.
Complexity issues of the packing coloring problem were adressed
in~\cite{A11,ANT12,ANT14,FG10,G15,GHHHR08}.

Let $H$ be a subgraph of $G$. Since $d_G(u,v)\le d_H(u,v)$ for any two vertices $u,v\in V(H)$,
the restriction to $V(H)$ of any packing coloring of $G$ is a packing coloring
of $H$. This property obviously holds for digraphs as well.
Hence, having packing chromatic number at most $k$
is a hereditary property:

\begin{proposition}[Goddard, Hedetniemi, Hedetniemi, Harris and Rall~\cite{GHHHR08}]\mbox{}\\
Let $G$ and $H$ be two undirected graphs, or two digraphs.
If $H$ is a subgraph of $G$, then $\pcn(H)\le\pcn(G)$.
\label{prop:subgraph}
\end{proposition}

Fiala and Golovach~\cite{FG10} proved that determining the packing chromatic number is an NP-hard problem for undirected trees.
The exact value of the packing chromatic number of undirected trees with diameter at most 4 was given in~\cite{GHHHR08}.
The packing chromatic number of undirected paths and cycles has been
determined by Goddard {\em et al.}:

\begin{theorem}[Goddard, Hedetniemi, Hedetniemi, Harris and Rall~\cite{GHHHR08}]\mbox{}
\begin{enumerate}[label={\rm (\arabic*)}]
\item  For every $n\ge 1$, $\pcn(P_n)\le 3$. Moreover, $\pcn(P_n)=1$ if and only if $n=1$ and 
$\pcn(P_n)=2$ if and only if $n\in\{2,3\}$.
\item  For every $n\ge 3$, $3\le\pcn(C_n)\le 4$. Moreover, $\pcn(C_n)=3$ if and only if $n=3$ or $n\equiv 0\pmod 4$. 
\end{enumerate}
\label{th:goddard}
\end{theorem}

\medskip

In this paper, we  consider undirected graphs and {\em orientations}
of undirected graphs, obtained by giving to each edge of such a graph one of its two possible orientations.
The so-obtained {\em oriented graphs} are thus digraphs having no pair of opposite arcs.
Let $\OR{G}$ be any orientation of an undirected graph $G$. Since for any two vertices $u,v$ in $V(G)$
we have $d_G(u,v)\le d_{\OR{G}}(u,v)$, where $d_{\OR{G}}(u,v)$ denotes the directed distance between $u$ and $v$, we get: 

\begin{proposition}
For every orientation $\OR{G}$ of an undirected graph $G$, $\pcn(\OR{G})\le\pcn(G)$.
\label{prop:oriented}
\end{proposition}

Let $u$ be a vertex in an oriented graph $\OR{G}$. We say that $u$ is a \emph{source} if $u$ has no incoming arc
and that $u$ is a \emph{sink} if $u$ has no outgoing arc.
If $\OR{uvw}$ is a directed path in $\OR{G}$, then $d_{\OR{G}}(u,w)\le 2$. Hence, $u$ and $w$ cannot be both assigned color 2
in any packing coloring of~$\OR{G}$. From this observation, we get an easy characterization of oriented graphs with packing
chromatic number~2:

\begin{proposition}[La\"iche, Bouchemakh and Sopena~\cite{LBS15}]\mbox{}\\
For every orientation $\OR{G}$ of an undirected graph $G$, $\pcn(\OR{G})=2$
if and only if (i) $G$ is bipartite and (ii) one part of the bipartition of $G$
contains only sources or sinks in $\OR{G}$.
\label{prop:pcn2}
\end{proposition}

In~\cite{LBS15}, we determined the packing chromatic number of undirected and oriented generalized coronae of paths and cycles.
In particular, the packing chromatic number of oriented paths and cycles is given as follows: 

\begin{theorem}[La\"iche, Bouchemakh and Sopena~\cite{LBS15}]\mbox{}\\
Let $\OR{P_n}$ be any orientation of the path $P_n$. Then, for every $n\ge 2$,
$2\le \pcn(\OR{P_n}) \le 3$.
Moreover, $\pcn(\OR{P_n})=2$ if and only if 
one part of the bipartition of $P_n$
contains only sources or sinks in $\OR{P_n}$.
\label{th:or-Pn}
\end{theorem}

\begin{theorem}[La\"iche, Bouchemakh and Sopena~\cite{LBS15}]\mbox{}\\
Let $\OR{C_n}$ be any orientation of the cycle $C_n$. Then, for every $n\ge 3$,
$2\le \pcn(\OR{C_n}) \le 4$.
Moreover, 
$\pcn(\OR{C_n})=2$ if and only if $C_n$ is bipartite and one part of the bipartition 
contains only sources or sinks in $\OR{C_n}$,
and
$\pcn(\OR{C_n})=4$ if and only if $\OR{C_n}$ is a directed cycle (all arcs have the same direction), 
   $n\ge 5$ and $n\not\equiv 0\pmod 4$.
\label{th:or-Cn}
\end{theorem}

\medskip

The \emph{generalized theta graph} $\Theta_{\ell_1,\dots,\ell_p}$ is the graph obtained
by identifying the end-vertices of 
$p\ge 2$ 
paths with respective lengths $1\le \ell_1\le\dots\le \ell_p$. 
(Since we only consider simple graphs, note here that we necessarily have $\ell_2\ge 2$.)
Packing colorings of undirected generalized theta graphs were considered by William and Roy in~\cite{WR13}
who gave some necessary condition for such a graph to have packing chromatic number~4.
In this paper, we determine the packing chromatic number of every  
undirected generalized theta graph.

Our paper is organized as follows.
In Section~\ref{sec:undirected} we provide tight lower and upper bounds on the packing chromatic
number of undirected generalized theta graphs and characterize
undirected generalized theta graphs with any given packing chromatic number.
In Section~\ref{sec:oriented}, we provide tight lower and upper bounds on the
packing chromatic number of oriented generalized theta graphs.

\section{Undirected generalized theta graphs}
\label{sec:undirected}

In this section, we determine the packing chromatic number of undirected generalized theta graphs
$\Theta_{\ell_1,\dots,\ell_p}$. 
Since we only consider undirected graphs in this section, we will simply write 
generalized theta graph instead of undirected generalized theta graph.

In the rest of this paper, we  denote by $u$ and $v$ the end-vertices of the theta graph $\Theta_{\ell_1,\dots,\ell_p}$ and by $P_i=ux_i^1\dots x_i^{\ell_i-1}v$ the
corresponding paths of length $\ell_i$ for every $i$, $1\le i\le p$.
Moreover, we  denote by $n_\ell$, $\ell\ge 1$, the number of paths of length~$\ell$, that is
$$n_\ell=|\{i\ /\ 1\le i\le p,\ \ell_i=\ell\}|.$$

In order to describe $k$-colorings of paths, we  use \emph{color patterns}, given as words
on the alphabet $\{1,\dots,k\}$, using standard notation from Combinatorics on Words,
with $u^+=u^*u$ for every word $u$. Hence, for instance,
the color pattern $12(1312)^*4$ describes colorings of the form $124$, $1213124$, $1213121312\dots 4$.

We first prove the following general upper bound:

\begin{theorem}
For every  generalized theta graph $\Theta=\Theta_{\ell_1,\dots,\ell_p}$,
$p\ge 2$, 
$$\pcn(\Theta)\le \max\{5, n_3+2\}.$$ 
Moreover, this upper bound is tight whenever $n_3\ge 3$.
\label{th:und-theta}
\end{theorem}

\begin{proof}
We first prove that $\pcn(\Theta)\le 5$ whenever $n_3\le 3$.
Let $\varphi:V(\Theta)\longrightarrow\{1,\dots,5\}$ be the mapping defined as follows:
\begin{enumerate}
\item $\varphi(u)=4$, $\varphi(v)=5$,
\item the (at most three) paths of length $3$ are colored using the distinct patterns $4125$, $4215$ and $4315$,
\item if $\ell_i\equiv 0\pmod 4$, $\ell_i\ge 4$, $\varphi(P_i)$ is defined by the pattern 
$4121(3121)^*5$,
\item if $\ell_i\equiv 1\pmod 4$, $\ell_i\ge 5$, $\varphi(P_i)$ is defined by the pattern 
$41231(2131)^*5$,
\item if $\ell_i\equiv 2\pmod 4$, $\varphi(P_i)$ is defined by the pattern 
$41(2131)^*5$,
\item if $\ell_i\equiv 3\pmod 4$, $\ell_i\neq 3$, $\varphi(P_i)$ is defined by the pattern 
$412(3121)^*5$.
\end{enumerate}

We claim that $\varphi$ is a packing 5-coloring of $\Theta$.
To see that, we will show that for any two distinct vertices $x$ and $y$ with $\varphi(x)=\varphi(y)=c$,
$c\in\{1,2,3\}$, we have $d_\Theta(x,y)>c$ (the case $c\in\{4,5\}$ does not need to be considered
since there is only one vertex with color 4 and one vertex with color 5). Note first that the restriction
of $\varphi$ to any path $P_i$ is a packing coloring of $P_i$.
Hence, we just need to consider the case when $x$ and $y$ do not belong to the same path.
If $c=1$, the property obviously holds since only internal vertices are colored with color 1.
Since at most one vertex with color 2 is adjacent to $u$ and at most one vertex with color 2 is adjacent
to $v$, the property also holds when $c=2$.
Since at most one vertex with color 3 is adjacent to $u$, no vertex with color~3 is at
distance~2 from $u$ and no vertex with color 3 is adjacent 
to $v$, the property also holds when $c=3$. Hence, $\varphi$ is a packing 5-coloring of $\Theta$.

Finally, when $n_3>3$, we color three paths of length 3 as above and the remaining ones
using distinct patterns of the form 4165, 4175, etc. Since each color $c>5$ is used only once,
we clearly get a packing ($n_3+2$)-coloring of $\Theta$. 

The fact that $\max\{5, n_3+2\}$ is a tight upper bound whenever $n_3\ge 3$ follows from Lemma~\ref{lem:und-theta-only-3}
proven below.
\end{proof}

We will now characterize  generalized theta
graphs with packing chromatic number $k$ for every $k\ge 3$. 
Since every generalized theta graph contains a cycle,
we know by Proposition~\ref{prop:subgraph} and Theorem~\ref{th:goddard}(2)
 that $\pcn(\Theta)\ge 3$ for every  generalized theta graph $\Theta$.
 Moreover, Theorem~\ref{th:goddard}(2) characterizes  generalized theta
graphs $\Theta_{\ell_1,\dots,\ell_p}$ with packing chromatic number 3 and 4 whenever $p=2$.
Therefore, unless otherwise specified, we will always consider $p\ge 3$ in the rest of this section.

The next lemma determines the packing chromatic number of  generalized theta graphs
of the form $\Theta_{3,\dots,3}$:

\begin{lemma}
Let $\Theta=\Theta_{\ell_1,\dots,\ell_p}$, $p\ge 3$, with $n_3=p$.
We then have $\pcn(\Theta)=p+2$.
\label{lem:und-theta-only-3}
\end{lemma}

\begin{proof}
By Theorem~\ref{th:und-theta}, we have $\pcn(\Theta)\le p+2$.
Therefore, it is enough to prove that for every packing $k$-coloring $\pi$ of $\Theta$, $k\ge p+2$.

If $\pi(u)=\pi(v)=1$ then at most two remaining vertices can be assigned color 2 and
all other remaining vertices must be assigned distinct colors, so that $\pi$ uses at least
$2(p-1)+2=2p\ge p+2$ colors.

If $\pi(u)=1$ and $\pi(v)\neq 1$, then
none of the vertices $x_i^1$, $1\le i\le p$, can be assigned color 1 and,
since the $p+1$ vertices $\{v,x_1^1,\dots,x_{p}^1\}$ are pairwise at distance 2, they must be assigned distinct colors
so that $\pi$ must use at least $p+2$ colors. The case $\pi(v)=1$ and $\pi(u)\neq 1$ is similar.

Finally, if $\pi(u)\neq 1$ and $\pi(v)\neq 1$, then
at most $p$ internal vertices can be assigned color 1 (one per path).
Since any two internal vertices are at distance at most 3 from each other and from $u$ and $v$, and no three such
vertices (including $u$ and $v$) are pairwise at distance 3 from each other, color 2 can be used at most twice,
so that $\pi$ must use at least $p+2$ colors.
\end{proof}

The following lemma characterizes  generalized theta graphs with packing
chromatic number $k$ for every $k>5$:

\begin{lemma}
Let $\Theta=\Theta_{\ell_1,\dots,\ell_p}$, $p\ge 3$, 
be a  generalized theta graph. 
Then, for every $k>5$, $\pcn(\Theta)=k$ if and only if $n_3=k-2$.
\label{lem:und-pcn-greater-than-5}
\end{lemma}

\begin{proof}
If $n_3=k-2$, we get $\pcn(\Theta)\le k$ by Theorem~\ref{th:und-theta}, 
and $\pcn(\Theta)\ge k$ by Lemma~\ref{lem:und-theta-only-3}
and Proposition~\ref{prop:subgraph}.

If $\pcn(\Theta)=k$, we get $n_3\ge k-2$ by Theorem~\ref{th:und-theta} and  $n_3\le k-2$ by Lemma~\ref{lem:und-theta-only-3}.
\end{proof}


Generalized theta graphs
with packing chromatic number 3 are characterized as follows.

\begin{lemma}
Let $\Theta=\Theta_{\ell_1,\dots,\ell_p}$, $p\ge 2$, 
be a generalized theta graph.
We then have $\pcn(\Theta)=3$ if and only if 
one of the following conditions holds:
\begin{enumerate}[label={\rm (\roman*)}]
\item  $\ell_1=1$ and $\ell_2=\dots =\ell_p=2$, or
\item  for every $i$ and $j$,
$1\le i<j\le p$, $\ell_i+\ell_j\equiv 0\pmod 4$.
\end{enumerate}
\label{lem:und-pcn-3}
\end{lemma}

\begin{proof}
By Theorem~\ref{th:goddard}(2), if $p=2$ then $\pcn(\Theta)=3$ if and only if
$\ell_1=1$ and $\ell_2=2$, or $\ell_1+\ell_2\equiv 0\pmod 4$. Therefore, assume $p\ge 3$.

We first prove that if $\ell_1=1$, $\ell_2=2$ and $\ell_p>2$ then $\pcn(\Theta)>3$.
Assume to the contrary that there exists a packing 3-coloring $\pi$ of $\Theta$.
Since $P_1$ and $P_2$ induce a cycle of length 3, we necessarily have 
$\pi(x_2^1)=\pi(x_p^1)=1$ and, without loss of generality, $\pi(u)=2$ and $\pi(v)=3$,
which implies that no color is available for $x_p^2$ since 
$d_\Theta(x_p^2,x_p^1)=1$, $d_\Theta(x_p^2,u)=2$ and $d_\Theta(x_p^2,v)\le 3$, a contradiction.

We know by Theorem~\ref{th:goddard}(2) that, for every $n\ge 3$,
$3\le\pcn(C_n)\le 4$ and $\pcn(C_n)=3$ if and only if $n=3$ or $n\equiv 0\pmod 4$.
Therefore, if $\Theta$ contains a cycle of length $\ell\not\equiv 0\pmod 4$, $\ell>3$,
then $\pcn(\Theta)>3$. Clearly, this happens whenever there exist $i$ and $j$,
$1\le i<j\le p$, with $\ell_i+\ell_j=\ell$.

Conversely, assume first that $\ell_1=1$ and $\ell_i=2$ for every $i$, $2\le i\le p$.
In that case, a packing 3-coloring $\pi$ of $\Theta$
is obtained by coloring each path of length~2 with $213$.
Assume now that for every $i$ and $j$,
$1\le i<j\le p$, $\ell_i+\ell_j\equiv 0\pmod 4$.
We have two cases to consider.
If $\ell_i\equiv 0\pmod 4$ for every $i$, $1\le i\le p$, a packing 3-coloring $\pi$ of $\Theta$
is obtained by coloring each path $P_i$ with the color pattern $(2131)^*2$.
If $\ell_i\equiv 2\pmod 4$ for every $i$, $1\le i\le p$, a packing 3-coloring $\pi$ of $\Theta$
is obtained by coloring each path $P_i$ with the color pattern $21(3121)^*3$.

This completes the proof.
\end{proof}

It remains to characterize  generalized theta graphs
with packing chromatic number 4 and 5.
Thanks to Theorem~\ref{th:goddard}(2), we do not need to consider cycles. 
The following series of lemmas will allow us to characterize 
generalized theta graphs (assuming $p\ge 3$) with packing chromatic number at most 4, depending
on the colors assigned to the end-vertices $u$ and $v$.

The first three lemmas characterize generalized theta graphs that admit a packing 4-coloring $\pi$
with $\pi(u)=\pi(v)=4$, 3 or 2.


\begin{lemma}
Let $\Theta=\Theta_{\ell_1,\dots,\ell_p}$, $p\ge 3$, 
be a generalized theta graph.
There exists a packing $4$-coloring $\pi$ of $\Theta$ with $\pi(u)=\pi(v)=4$
if and only if $n_1=n_2=n_3=n_4=0$.
\label{lem:und-theta-44}
\end{lemma}

\begin{proof}
Suppose first that $\pi$ is a packing 4-coloring of $\Theta$ with $\pi(u)=\pi(v)=4$.
We then necessarily have $d(u,v)>4$, which implies $n_1=n_2=n_3=n_4=0$.

Conversely, suppose that $n_1=n_2=n_3=n_4=0$.
We can color each path $P_i$, $1\le i\le p$, of length $\ell_i\ge 5$, using the following patterns, 
depending on the value of $(\ell_i\mod 4)$:
\begin{itemize}
\item $4(1213)^+1214$, if $\ell_i\equiv 0\pmod 4$,
\item $413(1213)^*214$, if $\ell_i\equiv 1\pmod 4$,
\item $4(1213)^+14$, if $\ell_i\equiv 2\pmod 4$,
\item $4(1213)^+214$, if $\ell_i\equiv 3\pmod 4$.
\end{itemize}
The so-obtained 4-coloring is clearly a packing 4-coloring of $\Theta$.
\end{proof}

\begin{lemma}
Let $\Theta=\Theta_{\ell_1,\dots,\ell_p}$, $p\ge 3$, 
be a generalized theta graph.
There exists a packing $4$-coloring $\pi$ of $\Theta$ with $\pi(u)=\pi(v)=3$
if and only if $n_1=n_2=n_3=0$, $n_5\le 2$ and $n_5+n_6\le 4$.
\label{lem:und-theta-33}
\end{lemma}

\begin{proof}
Suppose first that $\pi$ is a packing 4-coloring of $\Theta$ with $\pi(u)=\pi(v)=3$.
We then necessarily have $d(u,v)>3$, which implies $n_1=n_2=n_3=0$.
Note that we can only use colors 1, 2 and 4 for coloring the internal vertices of each
path $P_i$ with $\ell_i\le 7$, $1\le i\le p$. 
Therefore, each coloring of a path of length 5 must use once the color 4,
which implies $n_5\le 2$, since otherwise we would have two
vertices with color 4 at distance at most 4 from each other. 
Similarly, a path of length 6 can only be colored 3121413, 3141213, 3121423, 3241213, 3124123, 3214213
or 3214123,
which implies $n_6\le 4$ (again, due to vertices with colour 4).
Moreover, we necessarily have $n_6\le 2$ whenever $n_5=2$ and $n_6\le 3$ whenever $n_5=1$,
which gives $n_5\le 2$ and $n_5+n_6\le 4$.

Conversely, suppose that $n_1=n_2=n_3=0$, $n_5\le 2$ and $n_5+n_6\le 4$.
We color each path of length~4 with 31213.
If $n_5=2$, we color the two paths
of length 5 with 312413 and 314213 and the (at most two) paths of length 6 with 3124123 and 3214213.
If $n_5=1$, we color the path of length 5 with 312413 and the (at most three)
paths of length 6 with 3141213, 3124123 and 3214213.
If $n_5=0$, we color the (at most four) paths of length 6 with 3121413, 3141213, 3124123 and 3214213.

Finally, we color each path $P_i$, $1\le i\le p$, of length $\ell_i\ge 7$, using the following patterns, 
depending on the value of $(\ell_i\mod 4)$:
\begin{itemize}
\item $3(1213)^+1213$, if $\ell_i\equiv 0\pmod 4$,
\item $3(1213)^+41213$, if $\ell_i\equiv 1\pmod 4$,
\item $3(1213)^+141213$, if $\ell_i\equiv 2\pmod 4$,
\item $3(1213)^*1241213$, if $\ell_i\equiv 3\pmod 4$.
\end{itemize}
The so-obtained 4-coloring is clearly a packing 4-coloring of $\Theta$.
\end{proof}

\begin{lemma}
Let $\Theta=\Theta_{\ell_1,\dots,\ell_p}$, $p\ge 3$, 
be a generalized theta graph.
There exists a packing $4$-coloring $\pi$ of $\Theta$ with $\pi(u)=\pi(v)=2$
if and only if $n_1=n_2=0$ and one of the following conditions holds:
\begin{enumerate}[label={\rm (\roman*)}]
\item $n_3=1$ and $n_5+n_6+n_7=0$, or 
\item $n_3=0$ and $n_5+n_6+n_7\le 2$.
\end{enumerate}
\label{lem:und-theta-22}
\end{lemma}

\begin{proof}
Suppose first that $\pi$ is a packing 4-coloring of $\Theta$ with $\pi(u)=\pi(v)=2$.
We then necessarily have $d(u,v)>2$, which implies $n_1=n_2=0$.
Note that we can only use colors 1, 3 and 4 for coloring the internal vertices of each
path $P_i$ with $\ell_i\le 5$, $1\le i\le p$. 
Therefore, 
a path of length 3 can only be colored either 2342 (or 2432), 2132 (or 2312) or 2142 (or 2412), which implies $n_3\le 2$. 

If $n_3=2$ then, without loss of generality, the two paths of length 3 are colored either 2132 and 2142,
or 2132 and 2412. In both cases, no path of length $\ell\ge 4$ can be colored since  only the color 1 is available
for the vertices at distance 1 and 2 from $v$. This implies $p=2$, contradicting the assumption $p\ge 3$.
 
If $n_3=1$, as observed above, the corresponding path of length 3 is either colored 
2342 (or 2432), 2132 (or 2312) or 2142 (or 2412).
In the former case (assume, without loss of generality, that the path is colored 2342), 
every other vertex at distance~1 or~2 must be colored~1, which implies 
that only one additional path of length 2 may occur, contrary to the assumption $p\ge 3$.
In the second case (assume, without loss of generality, that the path is colored 2132),
every other vertex at distance at most~2 from $v$ must be  colored 1 or 4, which implies $\sum_{\ell\ge 4}n_\ell\le 1$,
so that $p=2$, contrary to the assumption $p\ge 3$.
The corresponding path of length 3 is thus colored 2142 (or 2412), so that
every other vertex at distance at most 2 from $u$ or $v$ must then be colored 1 or~3.
There is no such coloring for a path of length 5,
only one such coloring for a path of length 6, namely 2314132
(but this coloring is not valid since color~4 would be used on two vertices
at distance~4 from each other),
and two such colorings for a path of length 7, up to symmetry, namely 23124132
and 21324132. Since each of these colorings uses color 3 on a neighbor of $u$ or $v$
we necessarily have $n_5+n_6+n_7\le 1$.
If $n_5+n_6+n_7=1$ then, again, no other path can be colored since  only the color 1 is available
for the vertices at distance 1 and 2 from $v$ (or $u$), which implies $p=2$, contradicting the assumption $p\ge 3$.
Therefore, $n_5+n_6+n_7=0$ and condition (i) is satisfied.

If $n_3=0$ then, as observed above, every path of length 5, 6 or 7 must contain a vertex with color~4
at distance at most~2 from $u$ or $v$ since $p\ge 3$.
Therefore, at most two such paths can occur, that is $n_5+n_6+n_7\le 2$,
and thus condition (ii) is satisfied.

Conversely, suppose that $n_1=n_2=0$.
If $n_3=1$ and $n_5+n_6+n_7=0$, we color the path of length 3 with 2142 and every path of length 4 with 21312.
We then color each path $P_i$, $1\le i\le p$, of length $\ell_i\ge 8$, using the following patterns, 
depending on the value of $(\ell_i\mod 4)$:
\begin{itemize}
\item $2(1312)^+1312$, if $\ell_i\equiv 0\pmod 4$,
\item $2(1312)^+41312$, if $\ell_i\equiv 1\pmod 4$,
\item $2(1312)^+121312$, if $\ell_i\equiv 2\pmod 4$, 
\item $2(1312)^+4121312$, if $\ell_i\equiv 3\pmod 4$.
\end{itemize}

If $n_3=0$ and $n_5+n_6+n_7\le 2$ we first color
each path $P_i$, $1\le i\le p$, of length $\ell_i$, $4\le \ell_i\le 7$, as follows:
\begin{itemize}
\item $21312$, if $\ell_i=4$,
\item $213412$ or $214312$, if $\ell_i=5$,
\item $2131412$ or $2141321$, if $\ell_i=6$, 
\item $21321412$ or $21412312$, if $\ell_i=7$.
\end{itemize}
Note that if $n_5+n_6+n_7=2$ the two corresponding paths must use the patterns $214\ldots 2$ and $2\ldots 412$
so that the distance between the two vertices with color 4 is at least 5.
We then color each path $P_i$, $1\le i\le p$, of length $\ell_i\ge 8$, 
depending on the value of $(\ell_i\mod 4)$ as in the previous case.

The so-obtained 4-coloring is clearly a packing 4-coloring of $\Theta$.
\end{proof}


The next three lemmas characterize generalized theta graphs that admit a packing 4-coloring $\pi$
with $\pi(u),\pi(v)\in\{2,3,4\}$, $\pi(u)\neq\pi(v)$.

\begin{lemma}
Let $\Theta=\Theta_{\ell_1,\dots,\ell_p}$, $p\ge 3$, 
be a generalized theta graph.
There exists a packing $4$-coloring $\pi$ of $\Theta$ with $\pi(u)=3$ and $\pi(v)=4$
if and only if one of the following conditions holds:
\begin{enumerate}[label={\rm (\roman*)}]
\item $n_1\le 1$, $n_3\le 2$ and $n_5=n_6=0$, or 
\item $n_1=0$, $n_3\le 2$ and ($n_5=0$ or $n_5+n_6\le 1$).
\end{enumerate}
\label{lem:und-theta-34}
\end{lemma}

\begin{proof}
Suppose first that $\pi$ is a packing 4-coloring of $\Theta$ with $\pi(u)=3$ and $\pi(v)=4$.

There are only two possible colorings of a path of length 3, namely 3124 and 3214, which implies
that we can have at most two such paths (otherwise, we would have two vertices with color 2 at distance 2
from each other).

Suppose first that $n_1=1$.
In that case, since every internal vertex of a path of length 5 or 6 is at distance at most~3
from $u$ and $v$, the only available colors for these vertices are 1 and 2, so that $n_5+n_6=0$
and condition (i) is satisfied.

Suppose now that $n_1=0$.
Since the only possible coloring of a path of length 5 is 312134, 
we necessarily have $n_5\le 1$.
Consider the possible colorings of a path of length~6.
Since color 4 can only be used on the neighbor of $u$, the four other internal
vertices must use color 3 and thus color 3 has to be used on a vertex at distance at most 2 from $v$.
This implies $n_6=0$ whenever $n_5=1$ and thus condition (ii) is satisfied.

We finally prove that for every generalized theta graph satisfying any of these conditions,
there exists a packing 4-coloring $\pi$ with $\pi(u)=3$ and $\pi(v)=4$.
Every path of length 2 can be colored 314 and every path of length 4 can be colored 31214.
If $n_3=1$ the path of length 3 can be colored 3124, and if $n_3=2$ the paths of length 3 can be colored 3124 and 3214.
If $n_1=0$ and $n_5=1$, the path of length 5 can be colored 312134.
If $n_1=0$ and $n_5=0$, the path of length 6 can be colored 3121314.

It remains to prove that every path $P_i$, $1\le i\le p$, of length $\ell_i\ge 7$ can be colored. This can be
done using the following patterns, depending on the value of $(\ell_i\mod 4)$:
\begin{itemize}
\item $3(1213)^+1214$ if $\ell_i\equiv 0\pmod 4$,
\item $31214312(1312)^*14$ if $\ell_i\equiv 1\pmod 4$,
\item $31214(1312)^+14$ if $\ell_i\equiv 2\pmod 4$,
\item $3(1213)^+214$ if $\ell_i\equiv 3\pmod 4$.
\end{itemize}
The so-obtained 4-coloring is clearly a packing 4-coloring of $\Theta$.
\end{proof}

\begin{lemma}
Let $\Theta=\Theta_{\ell_1,\dots,\ell_p}$, $p\ge 3$, 
be a generalized theta graph.
There exists a packing $4$-coloring $\pi$ of $\Theta$ with $\pi(u)=2$ and $\pi(v)=4$
if and only if one of the following conditions holds:
\begin{enumerate}[label={\rm (\roman*)}]
\item $n_1\le 1$ and $n_3=n_7=0$, or 
\item $n_1\le 1$, $n_3=n_4=0$, $n_7\le 1$ and $n_8=0$, or 
\item $n_1\le 1$, $n_3\le 1$ and $n_4=n_7=n_8=0$, or 
\item $n_1=n_2=n_3=0$ and $n_7\le 1$, or
\item $n_1=n_2=n_3=n_4=0$, $n_7=2$ and $n_8=0$, or
\item $n_1=n_2=0$, $n_3\le 1$, $n_4=0$ and $n_7+n_8\le 1$, or
\item $n_1=n_3=n_7=0$, or
\item $n_1=n_3=n_4=0$, $n_7\le 1$ and $n_8=0$, or
\item $n_1=0$, $n_3\le 1$ and $n_4=n_7=n_8=0$.
\end{enumerate}
\label{lem:und-theta-24}
\end{lemma}

\begin{proof}
Suppose first that $\pi$ is a packing 4-coloring of $\Theta$ with $\pi(u)=2$ and $\pi(v)=4$.

Since every path of length~3 can be colored either 2134 or 2314, we 
necessarily have $n_3\le 1$ (otherwise, we would have two vertices with color~3 at distance~3
from each other).
Moreover, every path of length~4 can be colored either 23124, 21314 or 21324.
If a path of length~4 is colored 23124, every other vertex at distance at most~2
from $u$ must be colored 1 or 4, which implies that at most one additional path may occur,
contradicting our assumption $p\ge 3$.
If a path of length~4 is colored 21314 or 21324, therefore
using color~3 at distance~2 from $u$ and $v$, we necessarily have $n_4=0$
whenever a path uses color 3 on a neighbor of $u$ or $v$ (thus, in particular if $n_3=1$).

Suppose that $n_1=1$ and consider the possible colorings of a path of length 7.
On its internal vertices, color 4 cannot be used, color 3 can  be used only twice,
color 2 can be used only once
 and color 1 can be used three times.
Therefore, the only possible colorings of a path of length 7 are 21312134 and 23121314.
We thus necessarily have $n_7\le 1$, and $n_7=0$ whenever $n_3=1$ or $n_4=1$, otherwise we would have two vertices with color 3
at distance 2 or 3 from each other.
Similarly, for the internal vertices of a path of length 8,
color 4 cannot be used, 
color 3 and 2 can both be used at most twice,
 and color 1 can be used at most four times.
Therefore, the only colorings of a path of length 8 are 213121314 and 231213214.
We thus necessarily have $n_8=0$ whenever $n_3=1$ or $n_7=1$, again because of vertices with color 3. 

Therefore, one of the conditions (i), (ii) or (iii) is satisfied.

Suppose now that $n_1=0$.
We already know that $n_3\le 1$, and that $n_4=0$ whenever $n_3=1$.
If $n_2\ge 1$, the coloring 214 can be used (the other possible coloring, namely 234,
cannot give a better result).
Now, the possible colorings of a path of length 7 are 21312134 and 21431214
(using the color 3 or 4 on the neighbour of $u$ cannot give a better coloring than these two colorings).
This implies $n_7\le 2$ (because of vertices with color 3 or 4) and
both these colorings must be used when $n_7=2$. Moreover, if $n_3=1$ then the coloring 21312134
cannot be used and thus $n_7\le 1$ in that case. On the other hand, the coloring 21431214
cannot be used whenever $n_2\ge 1$.
If $n_4\ge 1$, the coloring $21314$ can be used (the two other possible colorings,
namely 21324 and 23124 cannot give a better result) which implies that the coloring 
21312134 cannot be used and thus $n_7\le 1$.

Similarly, the possible colorings of a path of length 8 
are 213121314 and 214131214 (again, using the color 3 or 4 on the neighbour of $u$ cannot give a 
better coloring than these two colorings). 
If $n_2\ge 1$ or $n_7=2$, the coloring 214131214 cannot be used,
because of vertices with color 4.
On the other hand, the coloring 213121314 cannot be used whenever $n_3=1$, or $n_2\ge 1$ and $n_7=1$,
or $n_2=0$ and $n_7=2$, because of vertices with color 3.

We can summarize the case $n_1=0$ as follows.
We necessarily have $n_7\le 2$, $n_3\le 1$, and $n_4=0$ whenever $n_3=1$.
If $n_2\ge 1$, we must have $n_7\le 1$, and $n_3=n_4=n_8=0$ whenever $n_7=1$.
Moreover, if $n_7=0$, we must have $n_8=0$ whenever $n_3=1$.
Therefore, one of the conditions (vii), (viii) or (ix) must hold.
Suppose now that $n_2=0$. 
If $n_3=1$, we necessarily have $n_7+n_8\le 1$ and condition (vi) holds.
If $n_3=0$ and $n_7=2$, we necessarily have $n_4=n_8=0$ and condition (v) holds.
If $n_3=0$ and $n_7\le 1$ then condition (iv) holds.
Therefore, one of the conditions (iv) to (ix) must hold.

We finally prove that for every generalized theta graph satisfying any of these conditions,
there exists a packing 4-coloring $\pi$ with $\pi(u)=2$ and $\pi(v)=4$.
We first color all the paths $P_i$, $1\le i\le p$, of length $\ell_i\notin\{3,7,8\}$, if any, as follows:
\begin{itemize}
\item $\ell_i=2$: 214,
\item $\ell_i=4$: 21314,
\item $\ell_i=5$: 213214,
\item $\ell_i=6$: 2131214,
\item $\ell_i\ge 9$: for these paths, we use the following patterns, depending on the value of $(\ell_i\mod 4)$:
    \begin{itemize}
    \item $2(1312)^+14131214$ if $\ell_i\equiv 0\pmod 4$, 
    \item $2(1312)^+13214$ if $\ell_i\equiv 1\pmod 4$,
    \item $2(1312)^+14$ if $\ell_i\equiv 2\pmod 4$,
    \item $2(1312)^+1431214$ if $\ell_i\equiv 3\pmod 4$.
    \end{itemize}
\end{itemize}
It remains to color the paths of length 3, 7 or 8. This can be done according to the condition
of the lemma that is satisfied:
\begin{enumerate}[label={\rm (\roman*)}]
\item 
All the paths of length $\ell\neq 8$ are already colored.
The paths of length 8, if any, can be colored 213121314.
\item 
The path of length 7 is colored 21312134 (recall that we have no path of length 3 or 4 in that case).
\item 
The path of length 3 is colored 2134 (recall that we have no path of length 4 in that case).
\item 
The path of length 7, if any, is colored 21431214 
and all the paths of length 8 are colored 213121314 (recall that we have no path of length 1, 2 or 3 in that case).
\item 
The two paths of length 7 are colored 21312134 and 21431214 (recall that we have no path of length less than 5 in that case).
\item 
The path of length 3 is colored 2134, the path of length 7, if any, is colored 21431214
 and the path of length 8, if any, is colored 214131214
(recall that we have no path of length 1, 2 or 4 and at most one path of length either 7 or 8 in that case).
\item 
All the paths of length $\ell\neq 8$ are already colored.
The paths of length 8, if any, can be colored 213121314.
\item 
The path of length 7 is colored 21312134
(recall that we have no path of length 1, 3, 4 or 8 in that case).
\item 
The path of length 3 is colored 2134
(recall that we have no path of length 1, 4, 7 or 8 in that case).
\end{enumerate}
In all cases, the so-obtained 4-coloring is clearly a packing 4-coloring of $\Theta$.
\end{proof}

\begin{lemma}
Let $\Theta=\Theta_{\ell_1,\dots,\ell_p}$, $p\ge 3$, 
be a generalized theta graph.
There exists a packing $4$-coloring $\pi$ of $\Theta$ with $\pi(u)=2$ and $\pi(v)=3$
if and only if one of the following conditions holds:
\begin{enumerate}[label={\rm (\roman*)}]
\item $n_1\le 1$ and $\sum_{i\ge 3}n_i\le 1$, or 
\item $n_1=0$ and $n_3+n_4+n_5\le 1$.
\end{enumerate}
\label{lem:und-theta-23}
\end{lemma}

\begin{proof}
Suppose first that $\pi$ is a packing 4-coloring of $\Theta$ with $\pi(u)=2$ and $\pi(v)=3$.

If $n_1=1$ then all the neighbors of $u$ and $v$ must be colored 1 or 4.
In every path of length $\ell\ge 3$, the vertex at distance 2 from $u$ is then necessarily colored 4
if the neighbor of $u$ is colored 1, or 1 if the neighbor of $u$ is colored 4.
Hence, we can have at most one such path 
(otherwise, we would have two vertices with color~4 at distance at most~4
from each other), that is $\sum_{i\ge 3}n_i\le 1$ and condition (i) is satisfied.

Suppose now that $n_1=0$.
The possible colorings of a path of length~3 are 2143 or 2413,
and the possible colorings of a path of length~4 are 21413, 21423 or 24123.
Moreover, the possible colorings of a path of length~5 are 214123 or 241213,
or 231213 (note that 231413 cannot give a better result) whenever $n_2=0$.
Therefore, if $n_2=1$ we necessarily have $n_3+n_4+n_5\le 1$, because of color~4.
On the other hand, if $n_2=0$ we can have $n_5=1$ which still implies $n_3+n_4\le 1$.
But if $n_3+n_4=1$ then the coloring of the path of length~5 must use color~3
on the neighbor of $u$, so that only colors 1 or 4 can be used on vertices at distance
at most~2 from $u$ and thus only one such path can exists (which is the assumed existing
path of length 3 or 4), which contradicts the assumption $p\ge 3$.
Therefore, condition (ii) is satisfied.

We finally prove that for every generalized theta graph satisfying any of these conditions,
there exists a packing 4-coloring $\pi$ with $\pi(u)=2$ and $\pi(v)=3$.
Every path of length~2 can be colored 213. 
If there is a path of length 3 (which implies either $n_1=1$ and $\sum_{i\ge 4}n_i=0$,
or $n_1=n_4=n_5=0$), then we color this path with 2143.
If there is a path of length 4 (which implies either $n_1=1$, $n_3=0$ and $\sum_{i\ge 5}n_i=0$,
or $n_1=n_3=n_5=0$), then we color this path with 21413.
If there is a path of length 5 (which implies either $n_1=1$, $n_3=n_4=0$ and $\sum_{i\ge 6}n_i=0$,
or $n_1=n_3=n_4=0$), 
then we color this path with 214213.

It remains to prove that every path $P_i$ of length $\ell_i\ge 6$ can be colored. 
If $n_1=1$ then $n_3=n_4=n_5=0$ and we have only one such path.
We then color this path using one of the following patterns, depending on the value of $(\ell_i\mod 4)$:
\begin{itemize}
\item $214(1312)^+13$ if $\ell_i\equiv 0\pmod 4$,  
\item $2142(1312)^+13$ if $\ell_i\equiv 1\pmod 4$, 
\item $21412(1312)^*13$ if $\ell_i\equiv 2\pmod 4$, 
\item $214312(1312)^*13$ if $\ell_i\equiv 3\pmod 4$.
\end{itemize}
If $n_1=0$, we color any such path
using the following patterns, depending on the value of $(\ell_i\mod 4)$:
\begin{itemize}
\item $2(1312)^*13141213$ if $\ell_i\equiv 0\pmod 4$, 
\item $2(1312)^+41213$ if $\ell_i\equiv 1\pmod 4$,
\item $2(1312)^+13$ if $\ell_i\equiv 2\pmod 4$,
\item $2(1312)^*1341213$ if $\ell_i\equiv 3\pmod 4$.
\end{itemize}
The so-obtained 4-coloring is clearly a packing 4-coloring of $\Theta$.
\end{proof}


Using Lemmas~\ref{lem:und-theta-44} to~\ref{lem:und-theta-23}
we get a complete characterization of generalized theta graphs admitting
a packing 4-coloring that does not use color 1 on vertex $u$ nor on vertex $v$:

\begin{theorem}
Let $\Theta=\Theta_{\ell_1,\dots,\ell_p}$, $p\ge 3$, 
be a generalized theta graph.
There exists a packing $4$-coloring $\pi$ of $\Theta$ with $\pi(u)\neq 1$ and $\pi(v)\neq 1$
if and only if one of the following conditions holds:
\begin{enumerate}[label={\rm (\Alph*)}]
\item $n_1=n_2=n_3=n_4=0$, \label{xx_44} 
\item $n_1=n_2=n_3=0$, $n_5\le 2$ and $n_5+n_6\le 4$, \label{xx_33} 
\item $n_1=n_2=n_3=0$ and $n_7\le 1$, \label{xx_24iv} 
\item $n_1=n_2=n_3=0$ and $n_5+n_6+n_7\le 2$, \label{xx_22ii} 
\item $n_1=n_2=0$, $n_3\le 1$ and $n_5=n_6=n_7=0$, \label{xx_22i} 
\item $n_1=n_2=0$, $n_3\le 1$, $n_4=0$  and $n_7+n_8\le 1$, \label{xx_24vi} 
\item $n_1=n_3=n_4=0$, $n_7\le 1$ and $n_8=0$, \label{xx_24viii} 
\item $n_1=0$, $n_3\le 2$ and ($n_5=0$ or $n_5+n_6\le 1$), \label{xx_34ii} 
\item $n_1=0$ and $n_3+n_4+n_5\le 1$, \label{xx_23ii} 
\item $n_1\le 1$, $n_3\le 2$ and $n_5=n_6=0$, \label{xx_34i} 
\item $n_1\le 1$, and $n_3=n_7=0$, \label{xx_24i} 
\item $n_1\le 1$, $n_3=n_4=0$, $n_7\le 1$ and $n_8=0$, \label{xx_24ii} 
\item $n_1\le 1$, $n_3\le 1$ and $n_4=n_7=n_8=0$, \label{xx_24iii} 
\item $n_1\le 1$, and $\sum_{i\ge 3}n_i\le 1$. \label{xx_23i} 
\end{enumerate}
\label{th:und-theta-not1}
\end{theorem}

\begin{proof}
This theorem simply summarizes the results of Lemmas~\ref{lem:und-theta-44} to~\ref{lem:und-theta-23}:
\begin{itemize}
\item Item~\ref{xx_44} follows from Lemma~\ref{lem:und-theta-44} and contains case (v) of Lemma~\ref{lem:und-theta-24}.
\item Item~\ref{xx_33} follows from Lemma~\ref{lem:und-theta-33}.
\item Item~\ref{xx_24iv} follows from case (iv) of Lemma~\ref{lem:und-theta-24}.
\item Item~\ref{xx_22ii} follows from case (ii) of Lemma~\ref{lem:und-theta-22}.
\item Item~\ref{xx_22i} follows from case (i) of Lemma~\ref{lem:und-theta-22}.
\item Item~\ref{xx_24vi} follows from case (vi) of Lemma~\ref{lem:und-theta-24}.
\item Item~\ref{xx_24viii} follows from case (viii) of Lemma~\ref{lem:und-theta-24}.
\item Item~\ref{xx_34ii} follows from case (ii) of Lemma~\ref{lem:und-theta-34}.
\item Item~\ref{xx_23ii} follows from case (ii) of Lemma~\ref{lem:und-theta-23}.
\item Item~\ref{xx_34i} follows from case (i) of Lemma~\ref{lem:und-theta-34}.
\item Item~\ref{xx_24i} follows from case (i) of Lemma~\ref{lem:und-theta-24}
and contains (vii) of Lemma~\ref{lem:und-theta-24}.
\item Item~\ref{xx_24ii} follows from case (ii) of Lemma~\ref{lem:und-theta-24}.
\item Item~\ref{xx_24iii} follows from case (iii) Lemma~\ref{lem:und-theta-24}
and contains case (ix) of Lemma~\ref{lem:und-theta-24}.
\item Item~\ref{xx_23i} follows from case (i) of Lemma~\ref{lem:und-theta-23}.
\end{itemize}
Hence, all the cases have been considered, this concludes the proof.
\end{proof}

If a generalized theta graph $\Theta$ satisfies none of the conditions \ref{xx_44} to~\ref{xx_23i}
of Theorem~\ref{th:und-theta-not1},
then every packing 4-coloring of $\Theta$ must use color 1 on $u$ or $v$.
The following observation will be useful:

\begin{observation}
If a generalized theta graph $\Theta=\Theta_{\ell_1,\dots,\ell_p}$, 
$p\ge 3$, admits a packing $4$-coloring $\pi$ with $\pi(u)=1$ then we necessarily have $p=3$.
\label{obs:1-234}
\end{observation}

To see that, it suffices to note that no two neighbors of $u$ can be assigned the same
color and that none of them can be colored 1. 

The next lemma will show that no generalized theta graph satisfying 
none of the conditions \ref{xx_44} to \ref{xx_23i} admits a packing 4-coloring.
By Observation~\ref{obs:1-234}, it suffices to consider generalized theta graphs
of the form $\Theta_{\ell_1,\ell_2,\ell_3}$. Moreover, by symmetry, it suffices to
consider packing 4-colorings that assign the color 1 to $u$.


\begin{lemma}
If $\Theta=\Theta_{\ell_1,\ell_2,\ell_3}$
is a generalized theta graph and $\pi$ a 
packing $4$-coloring of $\Theta$ with $\pi(u)=1$,
then $\Theta$ satisfies at least one of the conditions \ref{xx_44} to \ref{xx_23i}.
\label{lem:no-11}
\end{lemma}

\begin{proof}
We consider two cases, according to the value of $n_1$.
\begin{enumerate}
\item $n_1=0$.\\
If $n_2=n_3=n_4=0$ then $\Theta$ satisfies condition~\ref{xx_44}. 

Observe that we cannot have $n_3=3$ since, by Lemma~\ref{lem:und-theta-only-3},
we would have $\pcn(\Theta)=5$, a contradiction.

If $n_3=2$ then we necessarily have $n_5+n_6\le 1$ and therefore $\Theta$ satisfies condition~\ref{xx_34ii}. 

If $n_3=1$ and $n_5=0$ then $\Theta$ satisfies condition~\ref{xx_34ii}. 
If $n_3=1$, $n_5=1$ and $n_6=0$ then, again, $\Theta$ satisfies condition~\ref{xx_34ii}. 
If $n_3=1$, $n_5=1$ and $n_6=1$ then, since we necessarily have $n_2=n_4=n_7=n_8=0$, $\Theta$ satisfies condition~\ref{xx_24vi}. 
If $n_3=1$ and $n_5=2$ then $\Theta$ also satisfies condition~\ref{xx_24vi}. 

Suppose that $n_3=0$.
If $n_4\ge 1$ and $n_2=0$ then we necessarily have $n_5\le 2$ and $n_5+n_6\le 4$
and  $\Theta$ satisfies condition~\ref{xx_33}. 
If $n_4\ge 1$ and $n_2\ge 1$ then we necessarily have $n_5+n_6\le 1$
and  $\Theta$ satisfies condition~\ref{xx_34ii}. 
If $n_4=0$ and $n_5\le 1$ then $\Theta$ satisfies condition~\ref{xx_23ii}. 
If $n_4=0$ and $n_5\ge 2$ then $\Theta$ satisfies condition~\ref{xx_24vi}  if $n_7=0$ or 
condition~\ref{xx_24viii}  if $n_7=1$ (since we then have $n_7+n_8\le 1$). 

\item $n_1=1$.\\
In that case, we necessarily have $n_3\le 2$.
If $n_3=2$ then we necessarily have $n_5=n_6=0$ and $\Theta$ satisfies condition~\ref{xx_34i}. 

If $n_3=1$ and $n_5=n_6=0$ then $\Theta$ satisfies condition~\ref{xx_34i}. 
If $n_3=1$ and $n_5+n_6=1$ then we necessarily have $n_4=n_7=n_8=0$ and $\Theta$ satisfies condition~\ref{xx_24iii}. 

Suppose that $n_3=0$.
If $n_7=0$ then $\Theta$ satisfies condition~\ref{xx_24i}. 
If $n_7=1$ and $n_4=n_8=0$
then $\Theta$ satisfies condition~\ref{xx_24ii}. 
If $n_7=n_4+n_8=1$ or $n_7=2$ then we necessarily have $n_5=n_6=0$ 
and  $\Theta$ satisfies condition~\ref{xx_34i}. 

\end{enumerate}
This completes the proof.
\end{proof}

We are now able to characterize  generalized theta graphs with packing
chromatic number~4:

\begin{theorem}
Let $\Theta=\Theta_{\ell_1,\dots,\ell_p}$, $p\ge 2$,
be a generalized theta graph.
We then have $\pcn(\Theta)=4$ if and only if either
\begin{enumerate}[label={\rm (\arabic*)}] 
\item $p=2$, $\ell_1+\ell_2\neq 3$ and $\ell_1+\ell_2\not\equiv 0\pmod 4$,
or 
\item $p\ge 3$, $n_2\neq p$, there exist $i_1,i_2$, $1\le i_1<i_2\le p$, 
such that $\ell_{i_1}+\ell_{i_2}\not\equiv 0\pmod 4$,
and $\Theta$ satisfies one of the conditions \ref{xx_44} to~\ref{xx_23i}.
\end{enumerate}
\label{th:und-pcn-4}
\end{theorem}

\begin{proof}
If $p=2$ the result follows from Theorem~\ref{th:goddard}(2).
If $p\ge 3$, the result follows from Theorem~\ref{th:und-theta-not1} and Lemma~\ref{lem:no-11}.
\end{proof}

Using Lemma~\ref{lem:und-pcn-greater-than-5}, Lemma~\ref{lem:und-pcn-3} and
Theorem~\ref{th:und-pcn-4}, we get that the packing chromatic number of any
 generalized theta graph $\Theta=\Theta_{\ell_1,\dots,\ell_p}$
can be computed in time $O(p)$.

\section{Oriented generalized theta graphs}
\label{sec:oriented}

In this section, we study the packing chromatic number of oriented generalized theta graphs
$\OR{\Theta}_{\ell_1,\dots,\ell_p}$. 
Recall that we  denote by $n_\ell$, $\ell\ge 1$, the number of paths of length~$\ell$, that is
$$n_\ell=|\{i\ /\ 1\le i\le p,\ \ell_i=\ell\}|.$$

We prove the following:

\begin{theorem}
For every oriented generalized theta graph $\OR{\Theta}=\ORT_{\ell_1,\dots,\ell_p}$, $p\ge 2$,
$2\le \pcn(\ORT)\le 5$ and these two bounds are tight.
\label{th:or-theta}
\end{theorem}

\begin{figure}
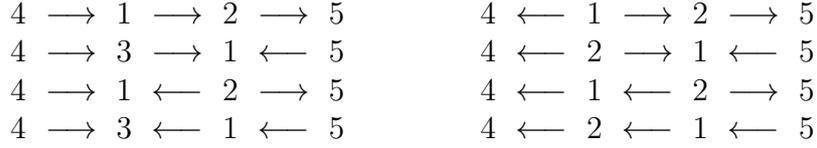

\begin{center}
$4\RR 1\RR 2\RR 5 \ \ \ \ \ \ \ \  \ \ \ \ \  4\LL 1\RR 2\RR 5$

$4\RR 3\RR 1\LL 5 \ \ \ \ \ \ \ \  \ \ \ \ \  4\LL 2\RR 1\LL 5$

$4\RR 1\LL 2\RR 5 \ \ \ \ \ \ \ \  \ \ \ \ \  4\LL 1\LL 2\RR 5$

$4\RR 3\LL 1\LL 5 \ \ \ \ \ \ \ \  \ \ \ \ \  4\LL 2\LL 1\LL 5$
\caption{Coloring of oriented paths of length 3 (proof of Theorem~\ref{th:or-theta})}
\label{fig:or-theta}
\end{center}
\end{figure}

\begin{proof}
It follows from Proposition~\ref{prop:pcn2} 
that 2 is a tight lower bound for $\pcn(\ORT)$.
By Proposition~\ref{prop:oriented} and Theorem~\ref{th:und-theta},
we know that $\pcn(\ORT)\le 5$ whenever $n_3\le 3$.

Assume thus that $n_3>3$.
Let us denote by $\OR{P_i}$ the orientation 
of the path $P_i$ for every $i$, $1\le i\le p$,
and let $\varphi:V(\ORT)\longrightarrow\{1,\dots,5\}$ be the mapping defined as 
in the proof of Theorem~\ref{th:und-theta}, except for the internal vertices
of the paths $\OR{P_i}$ with $\ell_i=3$, which are colored as shown in Figure~\ref{fig:or-theta}, according
to their orientation.

We claim that $\varphi$ is a packing 5-coloring of $\ORT$.
Again, the restriction
of $\varphi$ to any path $\OR{P_i}$ is a packing coloring of $\OR{P_i}$.
Moreover, from the proof of Theorem~\ref{th:und-theta}, we know that the restriction of $\varphi$
to $\bigcup\{\OR{P_i}:\ \ell_i\neq 3\}$ is a packing 5-coloring.
Hence, we just need to prove that for any two distinct vertices $x$ and $y$ with $x\in\OR{P_i}$, $\ell_i=3$, $\varphi(x)=\varphi(y)=c$,
$c\in\{2,3\}$ and $\{x,y\}\cap\{u,v\}=\emptyset$, we have $d_{\ORT}(x,y)>c$. 

Suppose first that $c=2$.
Since every vertex $y$ in $\OR{P_j}$, $\ell_j\neq 3$, with $\varphi(y)=2$ is at weak
directed distance at least 2 from $u$ and $v$, no conflict can occur between $x$ and $y$.
If $y$ belongs to some $\OR{P_j}$ with $\ell_j=3$ then 
the possible arcs are only $\OR{xu}$, $\OR{yu}$, $\OR{xv}$ and $\OR{yv}$ (see Figure ~\ref{fig:or-theta}),
and no conflict can occur between $x$ and $y$.

Suppose now that $c=3$. 
In that case, $x=x_i^1$ and $\OR{ux}$ is an arc (see Figure ~\ref{fig:or-theta}). 
Since every vertex $y$ in $\OR{P_j}$, $\ell_j\neq 3$, with $\varphi(y)=3$ is at weak
directed distance at least 3 from $u$ and at least 2 from $v$, no conflict can occur
between $x$ and such a $y$. If $y$ belongs to some $\OR{P_j}$ with $\ell_j=3$ then 
$y=y_i^1$ and $\OR{uy}$ is an arc, so that there is no conflict between $x$ and $y$.

We thus get $\pcn(\ORT)\le 5$.

Let us now prove that this bound is tight. For that, consider the oriented generalized
theta graph $\OR{\Theta_0}$ 
obtained by identifying (according
to their name, either $u$ or $v$) the end-vertices of the six following directed paths:
$$ux_1x_2x_3x_4v,\ \ uy_1y_2v,\ \ uz_1v,\ \ vx'_1x'_2x'_3x'_4u,\ \ vy'_1y'_2u,\ \ vz'_1u.$$
We claim that $\pcn(\OR{\Theta_0})=5$.
To see that, suppose that there exists a packing 4-coloring $\pi$ of $\OR{\Theta_0}$.
We consider five cases, according to
the values of $\pi(u)$ and $\pi(v)$ (up to symmetry).
Note that since $d_{\OR{\Theta_0}}(u,v)=2$, 
we necessarily have $\pi(u)=\pi(v)=1$ whenever $\pi(u)=\pi(v)$.

\begin{enumerate}
\item $\pi(u)=\pi(v)=1$.\\
In that case, no vertex in $\{y_1,y_2,y'_1,y'_2,z_1,z'_1\}$ can be colored 1.
Moreover, since any two vertices in $\{y_1,y_2,y'_1,y'_2\}$
are linked by a directed path (in either direction) of length at most 3,
we necessarily have either $\pi(y_1)=\pi(y'_1)=2$ and $\{\pi(y_2),\pi(y'_2)\}=\{3,4\}$
or $\pi(y_2)=\pi(y'_2)=2$ and $\{\pi(y_1),\pi(y'_1)\}=\{3,4\}$.
In both cases, one vertex in $\{z_1,z'_1\}$ cannot be colored.

\item $\pi(u)=1$, $\pi(v)\in\{2,3\}$.\\
If $\pi(v)=2$ (respectively $\pi(v)=3$), we necessarily have
$\{\pi(z_1),\pi(z'_1)\}=\{3,4\}$ (respectively $\{\pi(z_1),\pi(z'_1)\}=\{2,4\}$).
If $\pi(z_1)=4$ (respectively $\pi(z'_1)=4$), 
then $\{\pi(y'_1),\pi(y'_2)\}=\{3\}$ (respectively $\{\pi(y_1),\pi(y_2)\}=\{3\}$), a contradiction.

\item $\pi(u)=1$, $\pi(v)=4$.\\
In that case, we necessarily have $\{\pi(z_1),\pi(z'_1)\}=\{2,3\}$,
which implies $\{\pi(y_1),\pi(y'_2)\}=\{2,3\}$ and $\pi(y_2)=\pi(y'_1)=1$.
If $\pi(z_1)=2$ then $\pi(x_1)=2$ and $\pi(x_2)=1$, so that $x_3$ cannot
be colored.
If $\pi(z'_1)=2$ then $\pi(x'_1)=2$ and $\pi(x'_2)=1$, so that $x'_3$ cannot
be colored.

\item $\pi(u)=2$, $\pi(v)\in\{3,4\}$.\\
If $\pi(v)=3$ (respectively $\pi(v)=4$), we necessarily have
$\{\pi(y_1),\pi(y_2),\pi(y'_1),\pi(y'_2)\}=\{1,4\}$ (respectively $\{\pi(y_1),\pi(y_2),\pi(y'_1),\pi(y'_2)\}=\{1,3\}$),
a contradiction since any two vertices in $\{y_1,y_2,y'_1,y'_2\}$ are linked  
by a directed path (in either direction) of length at most 3.

\item $\pi(u)=3$, $\pi(v)=4$.\\
Since each vertex $x_i$, $1\le i\le 4$, is linked by a directed path (in either direction)
of length at most 3 to $u$ and by a directed path (in either direction)
of length at most 4 to $v$, we necessarily have $\{\pi(x_1),\pi(x_2),\pi(x_3),\pi(x_4)\}=\{1,2\}$, a contradiction.

\end{enumerate}

Therefore, every packing coloring of an oriented generalized
theta graph containing $\OR{\Theta_0}$ as a subgraph must use 5 colors.
\end{proof}

By Proposition~\ref{prop:pcn2}, we know that for every oriented generalized theta graph $\ORT$,
$\pcn(\ORT)=2$ if and only if $\ORT$ is bipartite and one part of the bipartition contains
only sources or sinks. 
However, characterizing oriented generalized theta graphs with packing chromatic number 3, 4 or 5
seems to be not so easy and we leave it as an open question.

From Lemma~\ref{lem:und-theta-44}, we get that $\pcn(\ORT)\le 4$ whenever $\Theta$ does not
contain any path of length less than 5. 
However, we believe that this value~5 cannot be decreased to~4.
In other words, we conjecture that
there exist oriented generalized theta graphs with
packing chromatic number~5 containing no path of length less than~4.


\bigskip

\noindent {\bf Acknowledgement.}
Most of this work has been done while the first author was
visiting LaBRI, thanks to a seven-months PROFAS-B+ grant cofunded by the Algerian
and French governments.
The third author was partially supported by the Cluster of excellence CPU, from the
Investments for the future Programme IdEx Bordeaux (ANR-10-IDEX-03-02).


\end{document}